%% file: Article.tex
\documentclass[aps,prl,reprint, amssymb, amsmath]{revtex4-2}
\usepackage{blindtext}

\usepackage{bbm} 

\usepackage{amsthm}
\newtheorem{theorem}{Theorem}
\theoremstyle{remark}
\newtheorem*{remark}{Remark}

\usepackage{graphicx}
\graphicspath{ {./Images/} }

\PassOptionsToPackage{hyphens}{url}

\RequirePackage{hyperref} 
\hypersetup{
    pdftitle={Extending the Known Region of Nonlocal Boxes that Collapse Communication Complexity},
    bookmarks=true,
    pdfpagemode=UseOutlines,
    pdfauthor={Marc-Olivier Proulx, Anne Broadbent, Pierre Botteron},
    linkbordercolor = [rgb]{1, 1, 1},
    filebordercolor = [rgb]{1, 1, 1},
    citebordercolor = [rgb]{1, 1, 1},
    menubordercolor = [rgb]{1, 1, 1},
    urlbordercolor = [rgb]{1, 1, 1},
    colorlinks = true,
    linkcolor = [rgb]{0, 0.2, 0.5},
    citecolor = [rgb]{0, 0.35, 0},
     urlcolor = [rgb]{0.6, 0.2, 0.4},
     breaklinks = true
}

\newcommand{\1}{\mathbbm{1}}

\newcommand{\CC}{{\normalfont\texttt{CC}}}
\newcommand{\CHSH}{{\normalfont\texttt{CHSH}}}
\newcommand{\eg}{\emph{e.g.} }
\newcommand{\I}{\mathtt{I}}
\newcommand{\ie}{\emph{i.e.} }
\renewcommand{\iff}{\emph{if, and only if,} }

\newcommand{\Maj}{\mathtt{Maj}}
\newcommand{\NS}{\mathcal{N\!S}}
\renewcommand{\P}{\mathtt{P}}
\newcommand{\PP}{\mathbb{P}}
\newcommand{\PR}{\mathtt{PR}}
\newcommand{\PRbar}{\overline{\mathtt{PR}}}
\newcommand{\PRprime}{\mathtt{PR'}}
\newcommand{\PRprimebar}{\overline{\mathtt{PR'}}}
\newcommand{\Protocol}{\mathcal{P}}

\newcommand{\R}{\mathbb{R}}
\newcommand{\SR}{\mathtt{SR}}
\newcommand{\SRbar}{\overline{\mathtt{SR}}}
\newcommand{\wP}{\widetilde\P}

\usepackage{xcolor}
\definecolor{darkred}{RGB}{180, 0, 0}
\definecolor{blue7}{rgb}{0.26, 0.44, 0.66}
\usepackage{bclogo} 
\usepackage{pgf,tikz}
\usepackage{tikz-cd}  
\usetikzlibrary{calc}
\usepackage{mathrsfs}
\usetikzlibrary{arrows}
\usepackage{subfiles} 

\begin{document}

\title{Extending the Known Region of Nonlocal Boxes\\that Collapse Communication Complexity}

\author{Pierre Botteron${}^1$}
\email{pierre.botteron@math.univ-toulouse.fr}

\author{Anne Broadbent${}^2$}
\email{abroadbe@uottawa.ca}

\author{Marc-Olivier Proulx${}^2$}
\email{mprou026@uottawa.ca}

\affiliation{${}^1$Institut de Mathématiques de Toulouse, Université de Toulouse (Paul Sabatier), France;}
\affiliation{${}^2$Nexus for Quantum Technologies, University of Ottawa, Canada.}

\begin{abstract}
Non-signalling boxes ($\NS$) are theoretical resources
defined by the principle of no-faster-than-light communication. They generalize quantum correlations, and some of them are known to collapse communication complexity ($\CC$).
However, this collapse is strongly believed to be unachievable in Nature, so its study provides intuition on which theories are unrealistic.
In the present letter, we find a better sufficient condition for a nonlocal box to collapse $\CC$, thus extending the known collapsing region. In some slices of $\NS$, we show this condition coincides with an area outside of an ellipse.
\end{abstract}

\maketitle

Entanglement is a fascinating relation linking pairs of particles.
It was experimentally confirmed in the late twentieth century \cite{CS78, AGR82, HBD+15}, and
it has the striking property of \emph{nonlocality}: two entangled particles, although being very distantly separated, provide strongly correlated results when their state is measured, yet the result of those measurements could not be known ahead of time \cite{Bell64,CHSH69}.

Nevertheless, this powerful nonlocality described by quantum mechanics is limited by  Tsirelson's famous bound~\cite{Tsirelson80}.
It is then natural to wonder if there could exist a more general theory than quantum mechanics to accurately describe the world, with more powerful nonlocality than quantum entanglement.
To that end, the common framework is the one of \emph{nonlocal boxes} (NLB)~\cite{PR94}. An NLB is a theoretical tool that generalizes the notions of shared randomness, quantum correlation and non-signalling correlation.
	As drawn in FIG.~\ref{fig: CC game}, an NLB has two input bits and two outputs bits. Alice has access only to the left side, and Bob to the right side. Immediately after inputting $x$ in the box, Alice receives $a$, whether or not Bob has already inputted his bit $y$.
	More formally, an NLB is characterized by a conditional distribution $\P(a,b\,|\,x,y)$ that satisfies the non-signalling conditions \cite{Shannon61, PR94}:
		$\textstyle \sum_{\tilde b} \P\big(a,\tilde b\,\big|\,x,0\big) = \sum_{\tilde b} \P\big(a,\tilde b\,\big|\,x,1\big)$ and
		$\textstyle \sum_{\tilde a} \P\big(\tilde a,b\,\big|\,0,y\big) = \sum_{\tilde a} \P\big(\tilde a,b\,\big|\,1,y\big)$
	for all $a,b,x,y\in\{0,1\}$.
	Denote by $\NS$ the set of all NLBs, which is an $8$-dimensional convex set with finitely many extremal points \cite{BLM+05, GKW+18}.

	Among the most famous boxes, there is the $\PR$ box, introduced by Popescu and Rohrlich \cite{PR94}, taking value $1/2$ if $a\oplus b=x\,y$, and $0$ otherwise, where the symbol $\oplus$ denotes the sum modulo $2$. Note that this box is designed to perfectly win at the $\CHSH$ game \cite{CHSH69}.
	In this work, we will also use the $\PR'$ box, taking value $1/2$ if $a\oplus b=(x\oplus 1)\,(y\oplus 1)$, and $0$ otherwise, which perfectly wins at $\CHSH'$ \cite{Branciard11} (same game as $\CHSH$ but with the rule $a\oplus b=(x\oplus 1)\,(y\oplus 1)$).
	In addition, we will use the fully random box $\I$, taking value $1/4$ for all inputs and outputs, and the shared randomness box $\SR$, taking value $1/2$ if $a=b$, and $0$ otherwise, independently of the entries $x$ and $y$. Note that $\SR$ is nothing more than a shared random bit.
	A bar above a box means that the behavior is the opposite one: $\overline{\P}:=1-\P$. \looseness=-1

\begin{figure}
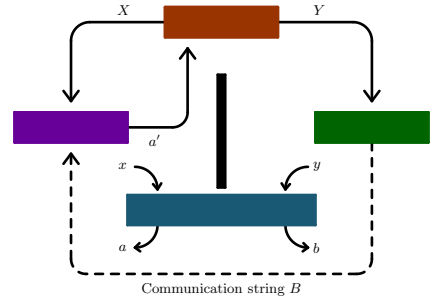

	\subfile{Communication_Complexity_Game}
	\caption{\emph{Communication complexity game.} (Colors online.) Lowercase letters $a,a'\!,b,x,y$ are bits, and capital letters are strings: $X\in\{0,1\}^n$, $Y\in\{0,1\}^m$ and $B\in\{0,1\}^k$.
	Let $f:\{0,1\}^n\times\{0,1\}^m\to\{0,1\}$ be known.
	Once the game starts, Alice and Bob are spacelike separated
	and the referee sends them the respective strings $X$ and $Y$.
	The goal is that Alice answers a bit $a'$ such that $a'=f(X,Y)$.
	To achieve it, Bob is allowed to send some \emph{communication bits} to Alice, but these bits are costly so he wants to send as few as possible.
	They may also use as many copies as they want of an NLB.
	}
	\label{fig: CC game}
\end{figure}

The notion of communication complexity ($\CC$) was introduced by Yao \cite{Yao79} and was widely studied in the late twentieth century \cite{KN96, RY20}. It can be viewed as a game as presented in FIG. \ref{fig: CC game}.
	The communication complexity $\CC(f)$ of $f$ is defined as the minimal amount of communication bits required to win at this game for any question strings $X$ and $Y$.
	One can see that $\CC(f)\leq m$ always, and that $\CC(f)\geq1$ if $f$ is not constant in $Y$.
	There exists also a probabilistic version of $\CC$ \cite{BBLMTU06}, in which Alice is allowed to make some mistakes: $\CC^p(f)$ (for some $p\in[0,1]$) is defined as the minimal amount of communication bits required to win with probability $\geq p$, for any $X$ and $Y$.
	Note that $\CC^{1/2}(f)=0$ for any $f$, using the strategy in which Alice always answers a uniformly random bit $a'\in\{0,1\}$.
	We say that $\CC$ \emph{collapses} (or that it is \emph{trivial}) when a single bit of communication is enough and that the error is bounded, \ie when there exists $p>1/2$ such that for all $f$ we have $\CC^p(f)\leq1$. This is strongly believed to be impossible in Nature since it would imply the absurdity that a single bit of communication is sufficient to distantly compute any $f$ \cite{vD99, BBLMTU06, BS09, BG15}.

\begin{figure*}
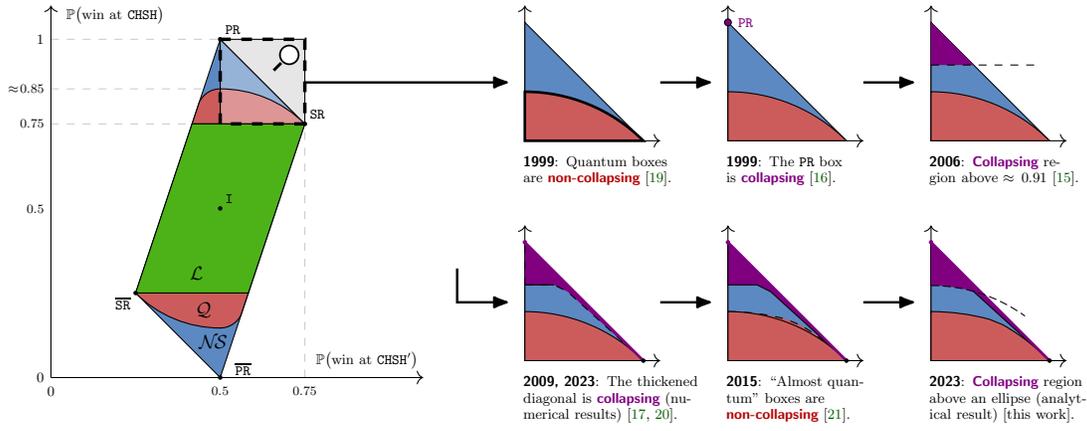

	\subfile{Historical_overview}
	\caption{\emph{Historical overview of collapsing boxes}, drawn in the slice of $\NS$ passing through $\PR$ and $\SR$ and $\I$. (Colors online.) In \textbf{\sffamily\color{darkred}red} and \textbf{\sffamily\color{violet}purple} are represented respectively the non-collapsing and the collapsing boxes.
	In \textbf{\sffamily\color{blue7}blue} is drawn the region of boxes for which we do not know yet if they collapse communication complexity.
	See \cite{BCMW10, BM06, FWW09, Mor16, SWH20, NSSRRB22PRL, EWC22b, Popescu14, Karvonen21, NPA08} for impossibility results and others.
	}
	\label{fig: historical overview}
\end{figure*}

Thus, the study of such a collapse helps to understand why some correlations are not allowed in quantum mechanics.
In the past two decades, some boxes were shown to be \emph{collapsing}, \ie to collapse $\CC$, see FIG. \ref{fig: historical overview}. However, there is still a major open question: 
do the other NLBs also collapse communication complexity?
In the present article, after generalizing the BBLMTU protocol \cite{BBLMTU06} (named after the authors' initials),
 we find a new sufficient condition that analytically extends the region of collapsing boxes, thus partially answering the open question.

	\section{I. Protocols}
	\label{section Protocols}
	
We define by induction a sequence of protocols $(\Protocol_k)_{k\geq0}$ generalizing the BBLMTU protocol \cite{BBLMTU06}, the main difference being that we add local uniformity.
\vspace{0.1cm}

	\textbf{Local Uniformization.}
We say that a box $\P\in\NS$ is \emph{locally uniform} if on each player's side, the box always outputs uniformly random bits: $\P(a\,|\,x)=1/2$ and $\P(b\,|\,y)=1/2$ for any $a,b,x,y\in\{0,1\}$, where $\P(a\,|\,x):=\sum_{b} \P(a, b\,|\,x,y)$ is independent of $y$ by non-signalling, and similarly for $\P(b\,|\,y)$.
The local uniformity will be useful many times in later computations.
However, some boxes $\P$ are \emph{not} locally uniform, \eg $\P=\frac{\PR+\P_0}{2}\in\NS$ where $\P_0$ is the box that always answers $(0,0)$ independently of the entries $(x,y)$.
This is why Alice and Bob use a ``trick": from $\P$ and a shared random bit $r$, they simulate another box $\widetilde\P\in\NS$ by adding $r$ to the outputs of $\P$.
That way, the new box $\wP$ is indeed locally uniform, and
importantly it has the same \emph{bias} $\eta_{xy}$ as the initial box $\P$ for all $x,y$:
\[
	\P\big(a\oplus b=xy\,|\,x,y\big)
	\,=\,
	\wP\big(a'\oplus b'=xy\,|\,x,y\big)
	\,=\,
	{\textstyle\frac{1+\eta_{xy}}{2}}
	\,,
\]
where $\eta_{xy}\in[-1,1]$ is defined as $\eta_{xy}
	:=
	2\,\P\big(a\oplus b=xy\,|\,x,y\big) -1
	=
	2\,\sum_{c} \P(c, c\oplus xy\,|\,x,y) -1
	$.

	\textbf{Protocol $\Protocol_0$.}
Fix a Boolean function $f:\{0,1\}^n\times\{0,1\}^m\to\{0,1\}$ and  strings $X\in\{0,1\}^n$ and $Y\in\{0,1\}^m$. The goal of the protocol is to perform a  \emph{distributed computation} of $f$ \cite{vD99}, \ie to find bits $a,b\in\{0,1\}$ known by Alice and Bob respectively such that:
\begin{equation}  \label{eq: distributed computation}  \tag{$\star$}
	a\oplus b = f(X,Y)\,.
\end{equation}
Assume Alice and Bob share uniformly random variables $Z\in\{0,1\}^m$ and $r\in\{0,1\}$. Upon receiving her string~$X$, Alice produces a bit $a:=f(X,Z)\oplus r$. As for Bob, if he receives a string $Y$ that is equal to $Z$, then he sets $b:=r$; otherwise he generates a local random variable $r_B$ and sets $b:=r_B$.
Now, separating the cases $Y=Z$ and $Y\neq Z$, the distributed computation \eqref{eq: distributed computation} is achieved with probability
\[
	p_0
	\,:=\,
	\PP\big( ``\eqref{eq: distributed computation}" \big)
	\,=\,
	\frac{1}{2^m} + \frac12\left(1-\frac{1}{2^m}\right)
	\,=\,
	\frac12 + \frac{1}{2^{m+1}}
	\,>\,
	\frac12\,.
\]
Due to the shared random bit $r$, note that the bit  $a$ is locally uniform:
$\PP(a\,|\,X) = 1/2$ for all $a,X$,
and similarly for $b$.
In total, this protocol uses $m+1$ shared random bits.
\\

	\textbf{Protocol $\Protocol_1$.}
As in $\Protocol_0$, we fix $f$, $X$ and $Y$, and we try to obtain the distributed computation \eqref{eq: distributed computation} with a better probability $p_1>p_0$. To that end, we realize four steps.

(a) We use the protocol $\Protocol_0$ independently three times, and obtain three pairs $(a_1,b_1), (a_2, b_2), (a_3,b_3)$ such that:
\[
	a_i\oplus b_i =
	\left\{
	\begin{array}{cl}
		f(X,Y) & \text{with prob. $p_0$}\\
		f(X,Y)\oplus1 & \text{with prob. $1-p_0$}\,.
	\end{array}
	\right.
\]
for $i=1,2,3$.
Note that this is a repetition code that will be decoded in (b) using a majority vote.

(b) The majority function $\Maj:\{0,1\}^3\to\{0,1\}$ is the function that outputs the most-appearing bit in its entries, \ie $\Maj(\alpha,\beta,\gamma) = \1_{\alpha+\beta+\gamma\geq2}$, where $\1$ is the indicator function.
The equality
\begin{equation} \label{eq: step b}  \tag{$\star\star$}
	f(X,Y) = \Maj\Big(a_1\oplus b_1,\, a_2\oplus b_2,\, a_3\oplus b_3\Big)
\end{equation}
occurs \iff at least two of the equations ``$f(X,Y) = a_i\oplus b_i$" ($i=1,2,3$) hold.
Denote $e_i:= a_i\oplus b_i \oplus f(X,Y)$,
and notice that ``$e_i=0$" \iff ``$a_i\oplus b_i=f(X,Y)$" for fixed $i$,
so that Equation \eqref{eq: step b} is equivalent to ``$\Maj(e_1,e_2,e_3)=0$".
But the $e_i$'s are independent
and $\PP(e_i=\alpha) = p_0^{1-\alpha} (1-p_0)^{\alpha}$ for $\alpha=0,1$, so Equality \eqref{eq: step b} holds with probability
\begin{align*}
	\PP\big(``\eqref{eq: step b}")
	\,=\,&
	\sum_{\underset{\text{s.t. }\Maj(\alpha,\beta,\gamma)=0}{\alpha,\beta,\gamma\in\{0,1\}}}  \PP(e_1=\alpha)\,\PP(e_2=\beta)\,\PP(e_3=\gamma)
	\\
	\,=\,&
	\sum_{\underset{\text{s.t. }\Maj(\alpha,\beta,\gamma)=0}{\alpha,\beta,\gamma\in\{0,1\}}} p_0^{3-\alpha-\beta-\gamma} (1-p_0)^{\alpha+\beta+\gamma}\,.
\end{align*}

(c) Now, we try to distributively compute the majority function. Observe that
\begin{multline*}
	\Maj\big( a_1\oplus b_1,\, a_2\oplus b_2,\, a_3\oplus b_3  \big)
	\\
	\,=\,
	\Maj(a_1,a_2,a_3) \oplus \Maj(b_1, b_2, b_3) \oplus r_1\, s_1 \oplus r_2\, s_2\,,
\end{multline*}
where $r_1:=a_1\oplus a_2$ and $s_1:=b_2\oplus b_3$ and $r_2:=a_2\oplus a_3$ and $s_2:=b_1\oplus b_2$. To distributively compute the two products $r_js_j$ ($j=1,2$), Alice and Bob use two copies of their locally uniform box $\wP$, see FIG. \ref{fig: distributively compute the products}.
They obtain pairs of bits $(a_1',b_1')$ and $(a_2',b_2')$ such that $a_j'\oplus b_j' = r_js_j$ with
bias $\eta_{r_j, s_j}$.
Consider the events $E_{\alpha,\beta,\gamma}:=``e_1=\alpha, e_2=\beta, e_3=\gamma"$ and $F_{\delta, \varepsilon, \zeta, \theta}:= ``r_1=\delta, r_2=\varepsilon, s_1=\zeta, s_2=\theta"$ where the greek letters are in $\{0,1\}$. On the one hand, under $E_{\alpha,\beta,\gamma}$ and $F_{\delta,\varepsilon,\zeta,\theta}$, we see that the equality
\begin{equation}  \label{eq: sum of products}
	r_1\,s_1 \oplus r_2\,s_2
	\,=\,
	(a_1'\oplus b_1')\oplus (a_2'\oplus b_2')
\end{equation}
holds \iff both of the equations $``r_js_j=a_j'\oplus b_j'"$ hold ($j=1,2$), or that none of them hold (because errors cancel out: $1\oplus 1=0$). Hence this equality holds with a bias $\eta_{\delta, \zeta}\eta_{\varepsilon, \theta}$:
\begin{equation}  \label{eq1}
	\PP\big( ``\eqref{eq: sum of products}" \,|\, E_{\alpha\beta\gamma}, F_{\delta\varepsilon\zeta\theta} \big) = \frac{1+\eta_{\delta, \zeta}\eta_{\varepsilon, \theta}}{2}\,,
\end{equation}
(conditionally to knowing $X$ and $Y$ as well).
On the other hand, seeing that the definitions of $r_j$ and $s_j$ lead to the relations $s_1=r_2\oplus e_2\oplus e_3$ and $s_2=r_1\oplus e_1\oplus e_2$, and using the independence of the $a_i$'s and their local uniform distribution in $\Protocol_0$, direct computations yield that:
\begin{equation}  \label{eq2}
	\PP(F_{\delta,\varepsilon,\zeta,\theta}\,|\,E_{\alpha,\beta,\gamma})
	\,=\,
	{\textstyle\frac{1}{4}}\, \1_{\zeta = \beta \oplus \gamma \oplus \varepsilon}\, \1_{\theta = \alpha \oplus \beta \oplus \delta}\,.
\end{equation}
Therefore, summing the products of \eqref{eq1} and \eqref{eq2} over all $\delta,\varepsilon,\zeta,\theta\in\{0,1\}$, we obtain:
\begin{align}
	\PP\big(``\eqref{eq: sum of products}" \,|\, E_{\alpha\beta\gamma}\big)
	=\,&   \label{eq: proba eq (1)}
	\hspace{-0.1cm}\sum_{\delta,\varepsilon\in\{0,1\}} \hspace{-0.cm}
	\frac{1+\eta_{\delta,\beta\oplus\gamma\oplus\varepsilon}\eta_{\varepsilon, \alpha\oplus\beta\oplus\delta}}{8}\,.
\end{align}
Hence, we obtain a distributed computation of the majority function as follows:
\begin{multline}   \tag{$\star\!\star\!\star$}   \label{eq: distributive computation of Maj}
	\Maj\big( a_1\oplus b_1,\, a_2\oplus b_2,\, a_3\oplus b_3  \big)\\
	\,=\,
	\scriptsize
	\underbrace{\Big( \Maj(a_1,a_2,a_3)\oplus a_1' \oplus a_2'\Big)}_{=:\,\widetilde a}
	\oplus \underbrace{\Big( \Maj(b_1,b_2,b_3)\oplus b_1' \oplus b_2'\Big)}_{=:\,\widetilde b}\,,
\end{multline}
with probability $\sum_{\delta,\varepsilon} \big(1+\eta_{\delta,\beta\oplus\gamma\oplus\varepsilon}\eta_{\varepsilon, \alpha\oplus\beta\oplus\delta}\big)/8$.
\begin{figure}[h]
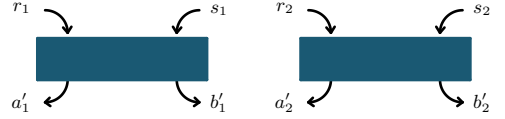

	\subfile{Nonlocal_Box}
	\caption{Distributively compute the products $r_1s_1$ and $r_2s_2$ with probability bias $\eta_{r_1, s_1}$ and $\eta_{r_2, s_2}$ respectively.}
	\label{fig: distributively compute the products}
\end{figure}

(d) Using steps (b) and (c), we obtain that the equality
\begin{equation}  \tag{$\star\!\star\!\star\star$}  \label{eq: 2nde distributed computation of f}
	f(X,Y) = \widetilde a\oplus \widetilde b
\end{equation}
holds \iff both \eqref{eq: step b} and \eqref{eq: distributive computation of Maj} hold, or that none of them hold (because errors cancel out: $1\oplus 1=0$). This happens with probability:
\begin{align*}
	p_1
	\,:=\,&
	\hspace{0.1cm}
	\small
	\PP\big( ``\eqref{eq: 2nde distributed computation of f}"\big)
	=
	\PP\big(\eqref{eq: step b} \land \eqref{eq: distributive computation of Maj}\big) + \PP\big(\neg\eqref{eq: step b} \land \neg\eqref{eq: distributive computation of Maj}\big)
	\\
	\,=\,&
	\hspace{-0.7cm}\sum_{\alpha,\beta,\gamma,\delta,\varepsilon\in\{0,1\}}\hspace{-0.7cm} \text{\scriptsize$
	p_0^{3\!-\!\alpha\!-\!\beta\!-\!\gamma} (1\!-\!p_0)^{\alpha\!+\!\beta\!+\!\gamma} \frac{1+ (\!-\!1)^{\Maj(\alpha,\beta,\gamma)}\eta_{\delta,\beta\oplus\gamma\oplus\varepsilon}\eta_{\varepsilon, \alpha\oplus\beta\oplus\delta}}{8}$}\,.
\end{align*}
where the sign ``$+$" from Equation \eqref{eq: proba eq (1)} was changed here into ``$(-1)^{\Maj(\alpha,\beta,\gamma)}$" because $\PP\big(\neg\eqref{eq: distributive computation of Maj}\big) = \sum_{\delta,\varepsilon} \big(1-\eta_{\delta,\beta\oplus\gamma\oplus\varepsilon}\eta_{\varepsilon, \alpha\oplus\beta\oplus\delta}\big)/8$, and this case exactly corresponds to the case where $\Maj(e_1,e_2,e_3)=1$.
Hence, we constructed a protocol $\Protocol_1$ based on $\Protocol_0$, and its probability of achieving \eqref{eq: distributed computation} is $p_1$. We will find in the next section a sufficient condition for which $p_1>p_0$.
In total, this protocol uses $3(m+2)-1$ shared random bits and $2$ copies of $\P$.
\\

\textbf{Protocol $\Protocol_{k+1}$ ($k\geq1$).} We proceed as in $\Protocol_1$: we build $\Protocol_{k+1}$ after performing $\Protocol_{k}$ three times.
In total, the protocol $\Protocol_{k+1}$ uses $3^{k+1}(m+2)-1$ shared random bits and $3^{k+1}-1$ copies of $\P$,
and it distributively computes $f$ with probability
\[
	p_{k+1}
	\,=\,
	\hspace{-0.8cm}\sum_{\alpha,\beta,\gamma,\delta,\varepsilon\in\{0,1\}}\hspace{-0.7cm} \text{\scriptsize$
	p_k^{3\!-\!\alpha\!-\!\beta\!-\!\gamma} \!(1\!\!-\!\!p_k)^{\alpha\!+\!\beta\!+\!\gamma} \frac{1+ (\!-\!1)^{\Maj(\alpha,\beta,\gamma)}\eta_{\delta,\beta\oplus\gamma\oplus\varepsilon}\eta_{\varepsilon, \alpha\oplus\beta\oplus\delta}}{8}$}\,.
\]

	\section{II. RESULT}
	
The probability bias associated to $p_{k+1}$ is $\mu_{k+1}:= 2\,p_{k+1}-1$ and it can be expressed as $\mu_{k+1} =F(\mu_k)$,
with:
\[
	F(\mu)
	\,:=\,
	\frac{\mu}{16} \,\left( A+B - \mu^2(A-B) \right)\,,
\]
where $A:= (\eta_{0,0} + \eta_{0,1} + \eta_{1,0} + \eta_{1,1})^2$
and $B:=2\eta_{0,0}^2 + 4\eta_{0,1}\eta_{1,0} + 2\eta_{1,1}^2$,
where $\eta_{xy}$ was introduced above as the probability bias of the box $\P$.
Note that $0\leq A\leq16$ and $-8\leq B \leq 8$ because $|\eta_{x,y}|\leq1$ for all $x,y$.

\begin{theorem}[Sufficient Condition]   \label{thm: sufficient condition}
	Nonlocal boxes for which $A+B>16$ collapse communication complexity.
\end{theorem}

\begin{proof}
	Assume $A+B>16$ ;
	this inequality has three consequences. (a) First, it gives $A-B>16-2B\geq0$ so that $F$ admits exactly three distinct \emph{fixed points} in $\R$:
	\[
		\left\{ 0,\, \pm \textstyle\sqrt{\frac{A+B-16}{A-B}} \right\}
		\,=:\,
		\big\{ 0,\, \pm \mu_* \big\}\,.
	\]
	
	(b) Second,
		as $\frac{dF}{d\mu}(\mu) = \frac{1}{16} \big(A+B - 3\mu^2(A-B)\big)$,
		the assumption implies that $F$ is increasing on $[-\mu_{\max{}}, \mu_{\max{}}]$, where $\mu_{\max{}}:=\sqrt{\frac{A+B}{3(A-B)}}>0$. Moreover, the assumption gives $\frac{\partial F}{\partial \mu}(0)>1$, so that the fixed point $0$ is repulsive.
		
	(c) Finally, as $A+B\leq 24$, we have $A+B-16 \leq \frac{A+B}{3}$. Therefore $\mu_*$ is smaller than or equal to $\mu_{\max{}}$ and:
	$$[0,\mu_*]\subseteq [-\mu_{\max{}}, \mu_{\max{}}]\,.$$

	Now, let $\P\in\NS$ be a box satisfying $A+B>16$.
	We provide Alice and Bob with as many shared random bits and as many copies of $\P$ as they want.
	We show that there exists a constant $p>1/2$ such that any arbitrary Boolean function $f:\{0,1\}^{n}\times\{0,1\}^m\to\{0,1\}$ can be distributively computed by Alice and Bob with probability $\geq p$, which means that communication complexity collapses.
	Using Section \hyperlink{section Protocols}{I}, the protocol $\Protocol_0$ enables to distributively compute $f$ with probability $p_0=(1+1/2^m)/2$, \ie with bias $\mu_0 = 1/2^m>0$. Up to adding muted variables in the entries of $f$, we may assume that $m$ is large enough so that $\mu_0\in(0,\mu_*)$. Then, combining (a), (b) and (c), we get that the sequence $(\mu_k)_k$ converges to the fixed point $\mu_*>0$. We set $p:=(1+\mu_*/2)/2 >1/2$ (or replace $\mu_*/2$ by any choice of $\mu\in(0,\mu_*)$), and we know that there exists a protocol $\Protocol_k$ for some $k$ large enough such that the probability $p_k$ of correctly distributively computing $f$ satisfies $p_k>p$. Finally, note that $p$ does \emph{not} depend on $f$: it only depends on $\mu_*$, which only depends on the $\eta_{x,y}$'s, which themselves only depend on $\P$. Hence communication complexity collapses.
\end{proof}

		\section{III. Cases of interest}
		
\textbf{Case 1: $\PR-\PRprime-\I$.} We consider a box $\P$ that is in the slice of $\NS$ passing through both $\PR$ and $\PR'$ and $\I$, studied in \cite{Branciard11}. In this case $\eta_{0,0} = \eta_{1,1}$ and $\eta_{0,1}=\eta_{1,0}$, and the condition $A+B>16$ of the Theorem reads as
$
	\eta_{0,0}^2 + \eta_{0,0}\eta_{0,1} + \eta_{0,1}^2 > 2
$.
We make a change of coordinates using the bias of winning at $\CHSH$ $\sigma=(\eta_{0,0}+\eta_{0,1})/2$ and the one of winning at $\CHSH'$  $\sigma'=(-\eta_{0,0} + \eta_{0,1})/2$, and we obtain:
\[
	 \sigma^2 + \textstyle\frac13\,\sigma'^2  > \frac23 \,,
	 \quad\quad\text{or}\quad\quad
	 \textstyle\frac13\,\sigma^2 + \sigma'^2  > \frac23 \,,
\]
where the second equation holds by changing the role of $\sigma$ and $\sigma'$ in the first one
(indeed, we may do it because flipping bits $x$ and $y$ allows to go from $\CHSH$ to $\CHSH'$).
These equations give rise to the purple collapsing area drawn in FIG. \ref{fig: case 1} (a). Interestingly, on the vertical axis, we find the same result as in \cite{BBLMTU06}: taking $\sigma'=0$, it is enough to have $\sigma>\sqrt{2/3}$, \ie to win at $\CHSH$ with probability $\frac{1+\sigma}{2}>\frac{3+\sqrt{6}}{6}\approx 0.91$.

\vspace{0.3cm}
\textbf{Case 2: $\PR-\SR-\I$.} We consider a box $\P$ that is in the slice of $\NS$ passing through both $\PR$ and $\SR$ and $\I$, studied in \cite{BS09}. In this case $\eta_{0,0} = \eta_{0,1} = \eta_{1,0}$, and the condition $A+B>16$ of the Theorem reads as
$
	5\eta_{0,0}^2 + 2\eta_{0,0}\eta_{1,1} + \eta_{1,1}^2 > \frac{16}{3}
$.
We make a change of coordinates using $\sigma=(3\eta_{0,0}+\eta_{1,1})/4$ and  $\sigma'=(\eta_{0,0} - \eta_{1,1})/4$, and we obtain:
\[
	 \sigma^2 + \sigma'^2  > \textstyle\frac23 \,.
\]
The induced collapsing area is represented in FIG. \ref{fig: case 1} (b). The same results also holds if we replace $\SR$ by any convex combination of $\P_0$ and $\P_1$, which are the boxes that always output respectively $(0,0)$ and $(1,1)$ independently of the entries $(x,y)$.

\definecolor{darkgreen}{rgb}{0.3, 0.6, 0.09}

\begin{figure}[h]
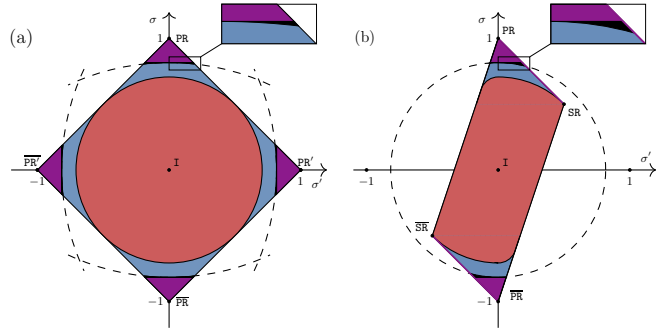

	\subfile{Case_1}
	\caption{\emph{Two slices of $\NS$} (colors online).
	In \textbf{\sffamily\color{violet} purple} is drawn the prior (analytically) known collapsing region.
	We extend it using Theorem \ref{thm: sufficient condition}: the \textbf{\sffamily\color{black} black} area is the new analytic collapsing region.
	The \textbf{\sffamily\color{darkred} red} area corresponds to the area of non-collapsing boxes.
	The \textbf{\sffamily\color{blue7}blue} area is the gap to be filled in red or purple (open problem).
	Drawings (a) and (b) represent the slices of $\NS$ passing through resp. $\{\PR, \PRprime,\I\}$ (case 1, finding interest in \cite{Branciard11}) and $\{\PR,\SR,\I\}$ (case 2, finding interest in \cite{BS09}).
	}
	\label{fig: case 1}
\end{figure}

\begin{remark}
	Even in comparison to previous numerical results, our protocol finds strictly new collapsing boxes. Indeed, for instance consider boxes in the black region of Figure~\ref{fig: case 1} that are close to the vertical axis: they are not distillable by means of the wirings of~\cite{BS09, EWC22PRL}, but our result shows that they are still collapsing.
\end{remark}

		\section{Conclusion}
		
After generalizing the BBLMTU protocol, we found in Theorem \ref{thm: sufficient condition} a new sufficient condition for a box to collapse communication complexity, with the following advantages:
(1) it is valid in the whole $8$-dimensional convex set $\NS$, in contrast with the analytical result of~\cite{BS09} (it holds only in the segment joining $\PR$ and $\SR$),  and
(2) it is completely analytical, with an explicit formula for the boundary of the new collapsing area, in contrast to previous numerical results~\cite{BS09, EWC22PRL} (as far as we know, the boundary of these two results has not yet been analytically computed).
In FIG. \ref{fig: case 1}, we presented two examples of new collapsing regions.
Note that the importance of our result is emphasized by considering the many known impossibility results \cite{BG15, Mor16, SWH20}.
According to our present intuition of Nature \cite{vD99, BBLMTU06, BS09, BG15}, a consequence is that our new collapsing boxes are unlikely to appear in Nature.

Hence, we partially answer the open question, but there is still a gap to be filled: what other nonlocal boxes collapse communication complexity?

		\section{Acknowledgements}
		
This work is based in part on~\cite{Marcothese}.
P.\,B.~is partly supported by the Institute for Quantum Technologies in Occitanie;
M.-O.\,P.~is supported by the Fonds de recherche du Québec --- Nature et technologies (FRQNT).
We acknowledge the support of the Natural Sciences and Engineering Research Council of Canada (NSERC) [funding Reference No. ALLRP/580876-2022 and No. RGPIN-2022-05167].

\bibliographystyle{apsrev4-2} 
\bibliography{Bibliography} 

\end{document}

%% file: Communication_Complexity_Game.tex
\definecolor{blue9}{rgb}{0.1, 0.35, 0.45}

\definecolor{qqttzz}{rgb}{0.,0.2,0.6}
\definecolor{rxsfyq}{rgb}{0.09019607843137255,0.1843137254901961,0.5019607843137255}
\definecolor{qqwuqq}{rgb}{0.,0.39215686274509803,0.}
\definecolor{wwqqzz}{rgb}{0.4,0.,0.6}
\definecolor{zzttqq}{rgb}{0.6,0.2,0.}

\newcommand{\EpaisseurTraits}{1pt}

\begin{tikzpicture}[line cap=round,line join=round,>=triangle 45,x=1.0cm,y=1.0cm, every node/.style={scale=0.6}, scale=0.5]

\fill[line width=2.pt,color=wwqqzz,fill=wwqqzz,fill opacity=0.10000000149011612] (-1.,1.) -- (-4.,1.) -- (-4.,0.2) -- (-1.,0.2) -- cycle;
\draw [line width=\EpaisseurTraits, color=wwqqzz] (-1.,1.)-- (-4.,1.);
\draw [line width=\EpaisseurTraits, color=wwqqzz] (-4.,1.)-- (-4.,0.2);
\draw [line width=\EpaisseurTraits, color=wwqqzz] (-4.,0.2)-- (-1.,0.2);
\draw [line width=\EpaisseurTraits, color=wwqqzz] (-1.,0.2)-- (-1.,1.);
\draw[color=wwqqzz] (-2.55, 0.6) node {\textbf{\sffamily Alice}};

\fill[line width=2.pt,color=qqwuqq,fill=qqwuqq,fill opacity=0.10000000149011612] (4.,1.) -- (7.,1.) -- (7.,0.2) -- (3.9957320768894165,0.20095767551444843) -- cycle;
\draw [line width=\EpaisseurTraits, color=qqwuqq] (4.,1.)-- (7.,1.);
\draw [line width=\EpaisseurTraits, color=qqwuqq] (7.,1.)-- (7.,0.2);
\draw [line width=\EpaisseurTraits, color=qqwuqq] (7.,0.2)-- (4,0.2);
\draw [line width=\EpaisseurTraits, color=qqwuqq] (4,0.2)-- (4.,1.);
\draw[color=qqwuqq] (5.5, 0.6) node {\textbf{\sffamily Bob}};

\fill[line width=2.pt,color=zzttqq,fill=zzttqq,fill opacity=0.10000000149011612] (0.,3.) -- (3.,3.) -- (3.,3.8) -- (0.,3.8) -- cycle;
\draw [line width=\EpaisseurTraits, color=zzttqq] (0.,3.)-- (3.,3.);
\draw [line width=\EpaisseurTraits, color=zzttqq] (3.,3.)-- (3.,3.8);
\draw [line width=\EpaisseurTraits, color=zzttqq] (3.,3.8)-- (0.,3.8);
\draw [line width=\EpaisseurTraits, color=zzttqq] (0.,3.8)-- (0.,3.);
\draw[color=zzttqq] (1.48, 3.4) node {\textbf{\sffamily Referee}};



\draw [line width=\EpaisseurTraits] (3.,3.4)-- (5.,3.4);
\draw [line width=\EpaisseurTraits] (5.5,2.9)-- (5.5,1.3);
\draw [line width=\EpaisseurTraits] (5.5, 1. + 0.3) -- (5.300185067839587,1.27 + 0.3);
\draw [line width=\EpaisseurTraits] (5.5, 1 + 0.3) -- (5.699814932160412,1.27 + 0.3);
\draw [shift={(5.,2.9)},line width=\EpaisseurTraits]  plot[domain=0.:1.5707963267948966,variable=\t]({1.*0.5*cos(\t r)+0.*0.5*sin(\t r)},{0.*0.5*cos(\t r)+1.*0.5*sin(\t r)});
\draw (4.1, 4) node[anchor=north] {$Y$};


\draw [line width=\EpaisseurTraits] (0.,3.4)-- (-2.,3.4);
\draw [line width=\EpaisseurTraits] (-2.5,2.9)-- (-2.5,1.3);
\draw [shift={(-2.,2.9)},line width=\EpaisseurTraits]  plot[domain=1.5707963267948966:3.141592653589793,variable=\t]({1.*0.5*cos(\t r)+0.*0.5*sin(\t r)},{0.*0.5*cos(\t r)+1.*0.5*sin(\t r)});
\draw [line width=\EpaisseurTraits] (-2.5,1 + 0.3) -- (-2.699814932160413,1.27 + 0.3);
\draw [line width=\EpaisseurTraits] (-2.5,1. + 0.3) -- (-2.300185067839587,1.27 + 0.3);
\draw (-1.1, 4) node[anchor=north] {$X$};

\fill[line width=2.pt,color=zzttqq,fill=zzttqq,fill opacity=0.10000000149011612] (0.,3.) -- (3.,3.) -- (3.,3.8) -- (0.,3.8) -- cycle;
\draw [line width=\EpaisseurTraits, color=zzttqq] (0.,3.)-- (3.,3.);
\draw [line width=\EpaisseurTraits, color=zzttqq] (3.,3.)-- (3.,3.8);
\draw [line width=\EpaisseurTraits, color=zzttqq] (3.,3.8)-- (0.,3.8);
\draw [line width=\EpaisseurTraits, color=zzttqq] (0.,3.8)-- (0.,3.);
\draw[color=zzttqq] (1.48, 3.4) node {\textbf{\sffamily Referee}};


\fill[line width=2.pt,fill=black,fill opacity=0.10000000149011612] (1.4,2.) -- (1.4,-1.) -- (1.6,-1.) -- (1.6,2.) -- cycle;
\draw [line width=\EpaisseurTraits] (1.4,2.)-- (1.4,-1.);
\draw [line width=\EpaisseurTraits] (1.4,-1.)-- (1.6,-1.);
\draw [line width=\EpaisseurTraits] (1.6,-1.)-- (1.6,2.);
\draw [line width=\EpaisseurTraits] (1.6,2.)-- (1.4,2.);



\draw [line width=\EpaisseurTraits] (-1.,0.6)-- (0.2,0.6);
\draw [line width=\EpaisseurTraits] (0.6,1.)-- (0.6,3 - 0.3);
\draw [line width=\EpaisseurTraits] (0.6, 3 - 0.3) -- (0.7998149321604123,2.73 - 0.3);
\draw [line width=\EpaisseurTraits] (0.6,3 - 0.3) -- (0.40018506783958707,2.73 - 0.3);
\draw [shift={(0.2,1.)},line width=\EpaisseurTraits]  plot[domain=-1.5707963267948966:0.,variable=\t]({1.*0.4*cos(\t r)+0.*0.4*sin(\t r)},{0.*0.4*cos(\t r)+1.*0.4*sin(\t r)});
\draw (-0.57, 0.55) node[anchor=north west] {$a'$};

\fill[line width=2.pt,color=wwqqzz,fill=wwqqzz,fill opacity=0.10000000149011612] (-1.,1.) -- (-4.,1.) -- (-4.,0.2) -- (-1.,0.2) -- cycle;
\draw [line width=\EpaisseurTraits, color=wwqqzz] (-1.,1.)-- (-4.,1.);
\draw [line width=\EpaisseurTraits, color=wwqqzz] (-4.,1.)-- (-4.,0.2);
\draw [line width=\EpaisseurTraits, color=wwqqzz] (-4.,0.2)-- (-1.,0.2);
\draw [line width=\EpaisseurTraits, color=wwqqzz] (-1.,0.2)-- (-1.,1.);
\draw[color=wwqqzz] (-2.55, 0.6) node {\textbf{\sffamily Alice}};



\newcommand{\mynumber}{0.15}

\draw [line width=\EpaisseurTraits] (-0.2,-1.2 + \mynumber)-- (-0.1,-1.04 + \mynumber);
\draw [line width=\EpaisseurTraits] (-0.32,-1.06 + \mynumber)-- (-0.2,-1.2 + \mynumber);
\draw [shift={(-0.8,-1.2)},line width=\EpaisseurTraits]  plot[domain=0.:1.5707963267948966,variable=\t]({1.*0.6*cos(\t r)+0.*0.6*sin(\t r)},{0.*0.6*cos(\t r)+1.*0.6*sin(\t r) + \mynumber});
\draw [line width=\EpaisseurTraits] (3.2,-1.2 + \mynumber)-- (3.1,-1.04 + \mynumber);
\draw [line width=\EpaisseurTraits] (3.2,-1.2 + \mynumber)-- (3.32,-1.06 + \mynumber);
\draw [shift={(3.8,-1.2)},line width=\EpaisseurTraits]  plot[domain=1.5707963267948966:3.141592653589793,variable=\t]({1.*0.6*cos(\t r)+0.*0.6*sin(\t r)},{0.*0.6*cos(\t r)+1.*0.6*sin(\t r) +  \mynumber});
\draw [line width=\EpaisseurTraits] (-0.8,-2.6)-- (-0.68,-2.48);
\draw [line width=\EpaisseurTraits] (-0.8,-2.6)-- (-0.66,-2.7);
\draw [shift={(3.8,-2.)},line width=\EpaisseurTraits]  plot[domain=3.141592653589793:4.71238898038469,variable=\t]({1.*0.6*cos(\t r)+0.*0.6*sin(\t r)},{0.*0.6*cos(\t r)+1.*0.6*sin(\t r)});
\draw [line width=\EpaisseurTraits] (3.8,-2.6)-- (3.68,-2.48);
\draw [line width=\EpaisseurTraits] (3.8,-2.6)-- (3.66,-2.7);
\draw [shift={(-0.8,-2.)},line width=\EpaisseurTraits]  plot[domain=-1.5707963267948966:0.,variable=\t]({1.*0.6*cos(\t r)+0.*0.6*sin(\t r)},{0.*0.6*cos(\t r)+1.*0.6*sin(\t r)});


\fill[line width=2.pt, color=blue9, fill=blue9, fill opacity=0.15000000596046448] (-1.,-2.) -- (4.,-2.) -- (4.,-1.2) -- (-1.,-1.2) -- cycle;
\draw [line width=\EpaisseurTraits, color=blue9] (-1.,-2.)-- (4.,-2.);
\draw [line width=\EpaisseurTraits, color=blue9] (4.,-2.)-- (4.,-1.2);
\draw [line width=\EpaisseurTraits, color=blue9] (4.,-1.2)-- (-1.,-1.2);
\draw [line width=\EpaisseurTraits, color=blue9] (-1.,-1.2)-- (-1.,-2.);



\draw (-0.88, -0.43) node[anchor=east] {$x$};
\draw (3.8, -0.43) node[anchor=west] {$y$};
\draw (-0.88, -2.6) node[anchor=east] {$a$};
\draw (3.8, -2.6) node[anchor=west] {$b$};

\draw [color=blue9](1.5, -1.6) node {\textbf{\sffamily Nonlocal box}};



\newcommand{\x}{-4.9}
\newcommand{\y}{-2.8}
\newcommand{\yy}{-1.1}

\draw [line width=\EpaisseurTraits, dashed] (5.5, 3.+\y)-- (5.5, 1.+\y+\yy);
\draw [line width=\EpaisseurTraits, dashed] (5.5-0.4, 0.6+\y+\yy)-- (2.8+\x, 0.6+\y+\yy);
\draw [line width=\EpaisseurTraits, dashed] (2.4+\x, 1.+\y+\yy)-- (2.4+\x, 3.- 0.3+\y);
\draw [shift={(5.5-0.4, 1.+\y+\yy)},line width=\EpaisseurTraits, dashed]  plot[domain=-1.5707963267948966 : 0.,variable=\t]({1.*0.4*cos(\t r)},{1.*0.4*sin(\t r)});
\draw [shift={(2.8+\x, 1.+\y+\yy)},line width=\EpaisseurTraits, dashed]  plot[domain=3.141592653589793 : 4.71238898038469,variable=\t]({1.*0.4*cos(\t r)},{1.*0.4*sin(\t r)});
\draw [line width=\EpaisseurTraits] (2.4+\x, 3- 0.3+\y) -- (2.599814932160412+\x, 2.73 - 0.3+\y);
\draw [line width=\EpaisseurTraits] (2.4+\x, 3 - 0.3+\y) -- (2.2001850678395867+\x, 2.73 - 0.3+\y);
\draw (1.5, 0.5+\y+\yy) node[anchor=north] {\small Communication string $B$};

\fill[line width=2.pt,color=qqwuqq,fill=qqwuqq,fill opacity=0.10000000149011612] (4.,1.) -- (7.,1.) -- (7.,0.2) -- (3.9957320768894165,0.20095767551444843) -- cycle;
\draw [line width=\EpaisseurTraits, color=qqwuqq] (4.,1.)-- (7.,1.);
\draw [line width=\EpaisseurTraits, color=qqwuqq] (7.,1.)-- (7.,0.2);
\draw [line width=\EpaisseurTraits, color=qqwuqq] (7.,0.2)-- (4,0.2);
\draw [line width=\EpaisseurTraits, color=qqwuqq] (4,0.2)-- (4.,1.);
\draw[color=qqwuqq] (5.5, 0.6) node {\textbf{\sffamily Bob}};

\end{tikzpicture}

%% file: Historical_overview.tex
\definecolor{qqttzz}{rgb}{0.,0.2,0.6}
\definecolor{rxsfyq}{rgb}{0.09019607843137255,0.1843137254901961,0.5019607843137255}
\definecolor{qqwuqq}{rgb}{0.,0.39215686274509803,0.}
\definecolor{wwqqzz}{rgb}{0.4,0.,0.6}
\definecolor{zzttqq}{rgb}{0.6,0.2,0.}

\definecolor{blue7}{rgb}{0.36, 0.54, 0.76}
\definecolor{red7}{rgb}{0.8, 0.36, 0.36}
\definecolor{green7}{rgb}{0.3, 0.7, 0.09}

\newcommand{\projectedWHITE}{\color{black!05}}
\newcommand{\projectedGREEN}{\color{green7!90}}

\begin{tikzpicture}[scale=0.5,every node/.style={scale=0.6}, scale=0.9]

\newcommand{\x}{gray!40}

\draw[dashed, color=\x] (0,10) -- (5,10) ;
\draw[dashed, color=\x] (0,{5+10*sqrt(1/2)/2}) -- (5,{5+10*sqrt(1/2)/2}) ; 
\draw[dashed, color=\x] (0,7.5) -- (4.17, 7.5) ;
\draw[dashed, color=\x] (7.5,0) -- (7.5, 7.5) ;

\draw[fill=blue7] (5,0) -- (7.5, 7.5) -- (5,10) -- (2.5, 2.5) -- (5,0) ;
\draw[fill = blue7!65] (7.5, 7.5) -- (5, 10) -- (5, 7.5) -- cycle;
\draw node at (4.8, 1.05) {\large$\NS$} ;

\coordinate (c1) at (5,5);
  \draw[fill=red7!65]
  ($(c1) + (45:35.355339059mm)$) arc (45:90:35.355339059mm)
  -- 
  (5, 7.5) -- cycle;

  \draw[fill=red7]
  ($(5,8) + (90:5.35mm)$) arc (90:155:6.8mm)
  -- 
  (4.17,7.5)
  --
  (5, 7.5) -- cycle;

  \draw[fill=red7]
  ($(c1) + (225:35.355339059mm)$) arc (225:270:35.355339059mm)
  -- 
  ($(5,2) + (270:5.35mm)$) arc (270:335:6.80mm)
  --
  (5.83,2.5) -- cycle;

\draw node at (4.5, 2) {\large$\mathcal{Q}$} ;

\draw[fill=green7] (2.5,2.5) -- (5.83,2.5) -- (7.5,7.5) -- (4.17,7.5) -- (2.5,2.5) ;
\draw node at (4.3, 3.1) {\large$\mathcal{L}$} ;

\draw[fill = gray!20] (7.5, 7.5) -- (7.5, 10) -- (5, 10) -- cycle;
\draw[dash pattern=on 4pt off 4pt, line width= 1.3pt] (5, 7.5) -- (5, 10) -- (7.5, 10) -- (7.5,7.5) -- (5, 7.5);
\draw (7.5, 10) node[below left] {$\bcloupe$};

\draw[->] (0,0) -- (0,11) ;
\draw node[anchor = west] at (0.3 ,10.7) {$\mathbb P\big(\text{win at $\CHSH$}\big)$};
\draw[->] (0,0) -- (11,0) ;
\draw node[anchor = east] at (11, 0.5) {$\mathbb P\big(\text{win at $\CHSH'$}\big)$};

\newcommand{\hauteur}{-0.12}
\newcommand{\largeur}{-0.05}

\draw node[left] at (\largeur,0) {$0$} ; 
\draw node[below] at (0, \hauteur) {$0$} ; 
\draw node[left] at (\largeur,5) {\footnotesize$0.5$} ; 
\draw node[below] at (5, \hauteur) {\footnotesize$0.5$} ; 
\draw node[left] at (\largeur,7.5) {\footnotesize$0.75$} ; 
\draw node[below] at (7.5, \hauteur) {\footnotesize$0.75$} ; 
\draw node[left] at (\largeur,{5+10*sqrt(1/2)/2}) {\footnotesize$\approx\!0.85$} ; 
\draw node[left] at (\largeur,10) {$1$} ;


\foreach \Point/\PointLabel 
in { (5, 10)/\PR
       }
\draw[fill=black] \Point circle (0.05) node[above right] {$\PointLabel$};

\draw[fill=black] (7.5, 7.5) circle (0.05) node[above right] {$\SR$};

\draw[fill=black] (5, 0) circle (0.05);
\draw[fill=black] (5.3, 0.25) node[right] {$\PRbar$};

\draw[fill=black] (2.5, 2.5) circle (0.05);
\draw[fill=black] (2.5, 2.5) node[below left] {$\SRbar$};

\draw[fill=black] (5, 5) circle (0.05) node[above right] {$\I$};


\renewcommand{\x}{12}
\newcommand{\y}{6}
\newcommand{\z}{1}

\draw[-{Latex[round]}, line width = 0.9pt] (7.5, 8.725) -- (\x + 1.5, 8.725);
\draw[->] (\x + 2, 6+\z) -- (\x + 6, 6+\z);
\draw[->] (\x + 2, 6+\z) -- (\x + 2, 10+\z);
\draw[text width = 3.8cm] (\x+1.8, 5.7+\z) node[below right] {\textbf{\sffamily 1999}:
	Quantum boxes are \textbf{\sffamily\color{darkred}non-collapsing} \cite{CvDNT99}.};

\draw[fill = blue7] (\x + 2, 6+\z) -- (\x+2, 9.5+ \z) -- (\x+2+3.5, 6+\z) -- cycle;
	
\coordinate (c1) at (\x+2, 6 - 3.5+\z);
  \draw[fill=red7, line width=1pt]
  ($(c1) + (90:49.5mm)$) arc (90:45:49.5mm)
  -- 
  (\x + 2, 6+\z) -- cycle;


\draw[-{Latex[round]}, line width = 0.9pt] (\x + \y, 7.725+\z) -- (\x + 1.5 + \y, 7.725+\z);
\draw[->] (\x + 2 + \y, 6+\z) -- (\x + 6 + \y, 6+\z);
\draw[->] (\x + 2 + \y, 6+\z) -- (\x + 2 + \y, 10+\z);
\draw[text width = 3.8cm] (\x+1.8 + \y, 5.7+\z) node[below right] {\textbf{\sffamily 1999}: The $\PR$ box
	is \textbf{\sffamily\color{violet}collapsing} \cite{vD99}.};

\draw[fill = blue7] (\x + 2+\y, 6+\z) -- (\x+2+\y, 9.5+ \z) -- (\x+2+3.5+\y, 6+\z) -- cycle;
	
\coordinate (c1) at (\x+2+\y, 6 - 3.5+\z);
  \draw[fill=red7]
  ($(c1) + (90:49.5mm)$) arc (90:45:49.5mm)
  -- 
  (\x + 2+\y, 6+\z) -- cycle;
  
\draw[fill = violet] (\x+2+\y, 9.5+ \z) circle (0.1) node[right] {\,\,\color{violet}$\PR$};


\draw[-{Latex[round]}, line width = 0.9pt] (\x + 2*\y, 7.725+\z) -- (\x + 1.5 + 2*\y, 7.725+\z);
\draw[->] (\x + 2 + 2*\y, 6+\z) -- (\x + 6 + 2*\y, 6+\z);
\draw[->] (\x + 2 + 2*\y, 6+\z) -- (\x + 2 + 2*\y, 10+\z);
\draw[text width = 3.8cm] (\x+1.8 + 2*\y, 5.7+\z) node[below right] {\textbf{\sffamily 2006}: \textbf{\sffamily\color{violet}Collapsing} region above $\approx0.91$ \cite{BBLMTU06}.};

\draw[fill = blue7] (\x + 2+2*\y, 6+\z) -- (\x+2+2*\y, 9.5+ \z) -- (\x+2+3.5+2*\y, 6+\z) -- cycle;
	
\coordinate (c1) at (\x+2+2*\y, 6 - 3.5+\z);
  \draw[fill=red7]
  ($(c1) + (90:49.5mm)$) arc (90:45:49.5mm)
  -- 
  (\x + 2+2*\y, 6+\z) -- cycle;

\draw[fill = violet, dash pattern = on 0pt off 1.5pt] (\x + 2+2*\y, 6+\z + 3.5*0.16/0.25) -- (\x+2+2*\y, 9.5+ \z) -- (\x+2+3.5+2*\y - 3.5*0.16/0.25, 6+\z + 3.5*0.16/0.25) -- cycle;

\draw[dash pattern = on 3pt off 3pt] (\x + 2+2*\y, 6+\z + 3.5*0.16/0.25) -- (\x+2+3.5+2*\y - 3.5*0.16/0.25 + 2, 6+\z + 3.5*0.16/0.25);


\renewcommand{\z}{1-6.5}

\draw[-{Latex[round]}, line width = 0.9pt] (\x, 8.725 + \z) -- (\x, 8.725 + \z -1) -- (\x + 1.5, 8.725 + \z -1);
\draw[->] (\x + 2, 6+\z) -- (\x + 6, 6+\z);
\draw[->] (\x + 2, 6+\z) -- (\x + 2, 10+\z);
\draw[text width = 3.9cm] (\x+1.8, 5.7+\z) node[below right] 
{\textbf{\sffamily 2009, 2023}: The thickened diagonal is \textbf{\sffamily\color{violet}collapsing} (numerical results) \cite{BS09,EWC22PRL}.};

\draw[fill = blue7] (\x + 2, 6+\z) -- (\x+2, 9.5+ \z) -- (\x+2+3.5, 6+\z) -- cycle;
	
\coordinate (c1) at (\x+2, 6 - 3.5+\z);
  \draw[fill=red7]
  ($(c1) + (90:49.5mm)$) arc (90:45:49.5mm)
  -- 
  (\x + 2, 6+\z) -- cycle;
  
\draw[fill = violet, dash pattern = on 3pt off 3pt] (\x + 2+0*\y, 6+\z + 3.5*0.16/0.25) -- (\x+2+0*\y, 9.5+ \z) --(\x+2+3.5, 6+\z) -- (\x+2+3.5+0*\y - 3.5*0.16/0.25 , 6+\z + 3.5*0.16/0.25 - 0.2) -- (\x+2+3.5+0*\y - 3.5*0.16/0.25 - 0.4, 6+\z + 3.5*0.16/0.25) -- cycle;
\draw[color = violet, line width = 1pt] (\x+2, 9.5+ \z) -- (\x+2+3.5, 6+\z);
\draw[fill = violet, color=violet] (\x+2+0*\y, 9.5+ \z) circle (0.05);
\draw[fill = black, color=black] (\x+2+3.5, 6+\z) circle (0.05);
\draw (\x+2+3.5+0*\y - 3.5*0.16/0.25 - 0.4, 6+\z + 3.5*0.16/0.25) -- (\x + 2+0*\y, 6+\z + 3.5*0.16/0.25);


\draw[-{Latex[round]}, line width = 0.9pt] (\x + \y, 7.725+\z) -- (\x + 1.5 + \y, 7.725+\z);
\draw[->] (\x + 2 + \y, 6+\z) -- (\x + 6 + \y, 6+\z);
\draw[->] (\x + 2 + \y, 6+\z) -- (\x + 2 + \y, 10+\z);
\draw[text width = 3.8cm] (\x+1.8 + \y, 5.7+\z) node[below right] {\textbf{\sffamily 2015}: ``Almost quantum" boxes are \textbf{\sffamily\color{darkred}non-collapsing} \cite{NGHA15}.};

\draw[fill = blue7] (\x + 2 + \y, 6+\z) -- (\x+2 + \y, 9.5+ \z) -- (\x+2+3.5 + \y, 6+\z) -- cycle;

\coordinate (c1) at (\x+2 + \y, 6 - 3.5+\z);
  \draw[fill=red7, dashed]
  ($(\x+2 + \y, 6 - 3.5+\z - 1.7) + (90:66.5mm)$) arc (90:75:66.5mm)
  --
   ($(\x+2 + \y - 1.2, 6 - 3.5+\z -1.2) + (60:66.5mm)$) arc (60:50:66.5mm)
  -- 
  (\x + 2 + \y, 6+\z) -- cycle;

\coordinate (c1) at (\x+2 + \y, 6 - 3.5+\z);
  \draw[fill=red7]
  ($(c1) + (90:49.5mm)$) arc (90:45:49.5mm)
  -- 
  (\x + 2 + \y, 6+\z) -- cycle;
  
\draw[fill = violet] (\x + 2+1*\y, 6+\z + 3.5*0.16/0.25) -- (\x+2+1*\y, 9.5+ \z) --(\x+2+3.5 + 1*\y, 6+\z) -- (\x+2+3.5+1*\y - 3.5*0.16/0.25 , 6+\z + 3.5*0.16/0.25 - 0.2) -- (\x+2+3.5+1*\y - 3.5*0.16/0.25 - 0.4, 6+\z + 3.5*0.16/0.25) -- cycle;
\draw[color = violet, line width = 1pt] (\x+2 + 1*\y, 9.5+ \z) -- (\x+2+3.5 + 1*\y, 6+\z);
\draw[fill = violet, color=violet] (\x+2+1*\y, 9.5+ \z) circle (0.05);
\draw[fill = black, color=black] (\x+2+3.5 + 1*\y, 6+\z) circle (0.05);


\draw[-{Latex[round]}, line width = 0.9pt] (\x + 2*\y, 7.725+\z) -- (\x + 1.5 + 2*\y, 7.725+\z);
\draw[->] (\x + 2 + 2*\y, 6+\z) -- (\x + 6 + 2*\y, 6+\z);
\draw[->] (\x + 2 + 2*\y, 6+\z) -- (\x + 2 + 2*\y, 10+\z);
\draw[text width = 4.1cm] (\x+1.8 + 2*\y, 5.7+\z) node[below right] {\textbf{\sffamily 2023}: \textbf{\sffamily\color{violet}Collapsing} region above an ellipse (analytical result) [this work].};

\draw[fill = blue7] (\x + 2 + 2*\y, 6+\z) -- (\x+2 +2* \y, 9.5+ \z) -- (\x+2+3.5 + 2*\y, 6+\z) -- cycle;
	
\coordinate (c1) at (\x+2 + \y, 6 - 3.5+\z);
  \draw[fill=red7]
  ($(\x+2 + 2*\y, 6 - 3.5+\z - 1.7) + (90:66.5mm)$) arc (90:75:66.5mm)
  --
   ($(\x+2 + 2*\y - 1.2, 6 - 3.5+\z -1.2) + (60:66.5mm)$) arc (60:50:66.5mm)
  --
   (\x+2+3.5 + 2*\y, 6+\z)
  -- 
  (\x + 2 + 2*\y, 6+\z) -- cycle;

\draw[fill = violet] (\x + 2+2*\y, 6+\z + 3.5*0.16/0.25) -- (\x+2+2*\y, 9.5+ \z) --(\x+2+3.5 + 2*\y, 6+\z) -- (\x+2+3.5+2*\y - 3.5*0.16/0.25 , 6+\z + 3.5*0.16/0.25 - 0.2) -- (\x+2+3.5+2*\y - 3.5*0.16/0.25 - 0.4, 6+\z + 3.5*0.16/0.25) -- cycle;
\draw[color = violet, line width = 1pt] (\x+2 + 2*\y, 9.5+ \z) -- (\x+2+3.5 + 2*\y, 6+\z);
\draw[fill = violet, color=violet] (\x+2+2*\y, 9.5+ \z) circle (0.05);
\draw[fill = black, color=black] (\x+2+3.5 + 2*\y, 6+\z) circle (0.05);

\coordinate (c1) at (\x+2 + 2*\y, 6 - 3.5+\z + 1);
  \draw[fill=violet, dash pattern = on 0pt off 1.5pt]
  ($(c1) + (90:47.2mm)$) arc (90:74:47.2mm)
  -- 
  (\x+2+2*\y, 9.5+ \z) -- cycle;

\draw[dash pattern = on 3pt off 1.5pt]
  ($(c1) + (90:47.2mm)$) arc (90:54:47.2mm) ;

\end{tikzpicture}

%% file: NonLocal_Box.tex
\definecolor{blue8}{rgb}{0.1, 0.35, 0.45}
\definecolor{qqttzz}{rgb}{0.,0.2,0.6}
\definecolor{rxsfyq}{rgb}{0.09019607843137255,0.1843137254901961,0.5019607843137255}
\definecolor{qqwuqq}{rgb}{0.,0.39215686274509803,0.}
\definecolor{wwqqzz}{rgb}{0.4,0.,0.6}
\definecolor{zzttqq}{rgb}{0.6,0.2,0.}

\begin{tikzpicture}[line cap=round,line join=round,>=triangle 45,x=1.0cm,y=1.0cm, scale=0.5, every node/.style={scale=0.8}]


\newcommand{\mynumber}{0.45}
\newcommand{\mywidth}{1}
\newcommand{\myhorizontalnumber}{-0.5}


\draw [line width=\mywidth pt] (-0.2,-1.2 + \mynumber)-- (-0.1,-1.04 + \mynumber);
\draw [line width=\mywidth pt] (-0.32,-1.06 + \mynumber)-- (-0.2,-1.2 + \mynumber);
\draw [shift={(-0.8,-1.2)},line width=\mywidth pt]  plot[domain=0.:1.5707963267948966,variable=\t]({1.*0.6*cos(\t r)+0.*0.6*sin(\t r)},{0.*0.6*cos(\t r)+1.*0.6*sin(\t r) + \mynumber});
\draw [line width=\mywidth pt] (3.2 + \myhorizontalnumber,-1.2 + \mynumber)-- (3.1 + \myhorizontalnumber,-1.04 + \mynumber);
\draw [line width=\mywidth pt] (3.2 + \myhorizontalnumber,-1.2 + \mynumber)-- (3.32 + \myhorizontalnumber,-1.06 + \mynumber);
\draw [shift={(3.8,-1.2)},line width=\mywidth pt]  plot[domain=1.5707963267948966:3.141592653589793,variable=\t]({1.*0.6*cos(\t r)+0.*0.6*sin(\t r) + \myhorizontalnumber},{0.*0.6*cos(\t r)+1.*0.6*sin(\t r) +  \mynumber});
\draw [line width=\mywidth pt] (-0.8,-2.6)-- (-0.68,-2.48);
\draw [line width=\mywidth pt] (-0.8,-2.6)-- (-0.66,-2.7);
\draw [shift={(-0.8,-2.)},line width=\mywidth pt]  plot[domain=-1.5707963267948966:0.,variable=\t]({1.*0.6*cos(\t r)+0.*0.6*sin(\t r)},{0.*0.6*cos(\t r)+1.*0.6*sin(\t r)});
\draw [line width=\mywidth pt] (3.8 + \myhorizontalnumber,-2.6)-- (3.68 + \myhorizontalnumber,-2.48);
\draw [line width=\mywidth pt] (3.8 + \myhorizontalnumber,-2.6)-- (3.66 + \myhorizontalnumber,-2.7);
\draw [shift={(3.8,-2.)},line width=\mywidth pt]  plot[domain=3.141592653589793:4.71238898038469,variable=\t]({1.*0.6*cos(\t r)+0.*0.6*sin(\t r) + \myhorizontalnumber},{0.*0.6*cos(\t r)+1.*0.6*sin(\t r)});

\fill[line width=\mywidth pt, color=blue8, fill=blue8, fill opacity=0.15000000596046448] (-1.,-2.) -- (4. + \myhorizontalnumber,-2.) -- (4. + \myhorizontalnumber,-0.9) -- (-1.,-0.9) -- cycle;
\draw [line width=1.2*\mywidth pt,color=blue8] (-1.,-2.)-- (4. + \myhorizontalnumber,-2.);
\draw [line width=1.2*\mywidth pt,color=blue8] (4. + \myhorizontalnumber,-2.)-- (4. + \myhorizontalnumber,-0.9);
\draw [line width=1.2*\mywidth pt,color=blue8] (4. + \myhorizontalnumber,-0.9)-- (-1.,-0.9);
\draw [line width=1.2*\mywidth pt,color=blue8] (-1.,-0.9)-- (-1.,-2.);

\draw (-1.0, -0.4+0.3) node[anchor=east] {$r_1$};
\draw (3.9 + \myhorizontalnumber, -0.43+0.3) node[anchor=west] {$s_1$};
\draw (-1.0, -2.6) node[anchor=east] {$a_1'$};
\draw (3.9 + \myhorizontalnumber, -2.6) node[anchor=west] {$b_1'$};

\draw [color=blue8](1.45 + \myhorizontalnumber/2, -1.45) node {\textbf{$\widetilde\P$}};


\newcommand{\shiftBox}{7}

\draw [line width=\mywidth pt] (-0.2 + \shiftBox,-1.2 + \mynumber)-- (-0.1 + \shiftBox,-1.04 + \mynumber);
\draw [line width=\mywidth pt] (-0.32 + \shiftBox,-1.06 + \mynumber)-- (-0.2 + \shiftBox,-1.2 + \mynumber);
\draw [shift={(-0.8,-1.2)},line width=\mywidth pt]  plot[domain=0.:1.5707963267948966,variable=\t]({1.*0.6*cos(\t r)+0.*0.6*sin(\t r) + \shiftBox},{0.*0.6*cos(\t r)+1.*0.6*sin(\t r) + \mynumber});
\draw [line width=\mywidth pt] (3.2 + \myhorizontalnumber + \shiftBox,-1.2 + \mynumber)-- (3.1 + \myhorizontalnumber + \shiftBox,-1.04 + \mynumber);
\draw [line width=\mywidth pt] (3.2 + \myhorizontalnumber + \shiftBox,-1.2 + \mynumber)-- (3.32 + \myhorizontalnumber + \shiftBox,-1.06 + \mynumber);
\draw [shift={(3.8,-1.2)},line width=\mywidth pt]  plot[domain=1.5707963267948966:3.141592653589793,variable=\t]({1.*0.6*cos(\t r)+0.*0.6*sin(\t r) + \myhorizontalnumber + \shiftBox},{0.*0.6*cos(\t r)+1.*0.6*sin(\t r) +  \mynumber});
\draw [line width=\mywidth pt] (-0.8 + \shiftBox,-2.6)-- (-0.68 + \shiftBox,-2.48);
\draw [line width=\mywidth pt] (-0.8 + \shiftBox,-2.6)-- (-0.66 + \shiftBox,-2.7);
\draw [shift={(-0.8,-2.)},line width=\mywidth pt]  plot[domain=-1.5707963267948966:0.,variable=\t]({1.*0.6*cos(\t r)+0.*0.6*sin(\t r) + \shiftBox},{0.*0.6*cos(\t r)+1.*0.6*sin(\t r)});
\draw [line width=\mywidth pt] (3.8 + \myhorizontalnumber + \shiftBox,-2.6)-- (3.68 + \myhorizontalnumber + \shiftBox,-2.48);
\draw [line width=\mywidth pt] (3.8 + \myhorizontalnumber + \shiftBox,-2.6)-- (3.66 + \myhorizontalnumber + \shiftBox,-2.7);
\draw [shift={(3.8,-2.)},line width=\mywidth pt]  plot[domain=3.141592653589793:4.71238898038469,variable=\t]({1.*0.6*cos(\t r)+0.*0.6*sin(\t r) + \myhorizontalnumber + \shiftBox},{0.*0.6*cos(\t r)+1.*0.6*sin(\t r)});

\fill[line width=\mywidth pt, color=blue8, fill=blue8, fill opacity=0.15000000596046448] (-1. + \shiftBox,-2.) -- (4. + \myhorizontalnumber + \shiftBox,-2.) -- (4. + \myhorizontalnumber + \shiftBox,-0.9) -- (-1. + \shiftBox,-0.9) -- cycle;
\draw [line width=1.2*\mywidth pt,color=blue8] (-1. + \shiftBox,-2.)-- (4. + \myhorizontalnumber + \shiftBox,-2.);
\draw [line width=1.2*\mywidth pt,color=blue8] (4. + \myhorizontalnumber + \shiftBox,-2.)-- (4. + \myhorizontalnumber + \shiftBox,-0.9);
\draw [line width=1.2*\mywidth pt,color=blue8] (4. + \myhorizontalnumber + \shiftBox,-0.9)-- (-1. + \shiftBox,-0.9);
\draw [line width=1.2*\mywidth pt,color=blue8] (-1. + \shiftBox,-0.9)-- (-1. + \shiftBox,-2.);

\draw (-1.0 + \shiftBox, -0.4+0.3) node[anchor=east] {$r_2$};
\draw (3.9 + \myhorizontalnumber + \shiftBox, -0.43+0.3) node[anchor=west] {$s_2$};
\draw (-1.0 + \shiftBox, -2.6) node[anchor=east] {$a_2'$};
\draw (3.9 + \myhorizontalnumber + \shiftBox, -2.6) node[anchor=west] {$b_2'$};

\draw [color=blue8](1.45 + \myhorizontalnumber/2 + \shiftBox, -1.45) node {\textbf{$\widetilde\P$}};

\end{tikzpicture}

%% file: Case_1.tex
\definecolor{qqttzz}{rgb}{0.,0.2,0.6}
\definecolor{rxsfyq}{rgb}{0.09019607843137255,0.1843137254901961,0.5019607843137255}
\definecolor{qqwuqq}{rgb}{0.,0.39215686274509803,0.}
\definecolor{wwqqzz}{rgb}{0.4,0.,0.6}
\definecolor{zzttqq}{rgb}{0.6,0.2,0.}

\definecolor{red7}{rgb}{0.8, 0.36, 0.36}
\definecolor{green7}{rgb}{0.3, 0.7, 0.09}

\newcommand{\projectedWHITE}{\color{black!05}}
\newcommand{\projectedGREEN}{\color{green7!90}}

\begin{tikzpicture}[scale=0.5,every node/.style={scale=0.5}, scale=0.7]

\draw node[left] at (0,10) {\Large(a)} ;

\draw[->] (5,-1) -- (5,11) ;
\draw node[anchor = east] at (4.7 ,10.7) {$\sigma$};
\draw[->] (-1,5) -- (11,5) ;
\draw node[anchor = east] at (11, 4.5) {$\sigma'$};

\newcommand{\hauteur}{5-0.12}
\newcommand{\largeur}{5-0.05}

\draw node[left] at (\largeur, 0) {$-1$} ; 
\draw node[below] at (0, \hauteur) {$-1$} ; 
\draw[fill=black] (0, 5) circle (0.05) ;
\draw node[left] at (\largeur,10) {$1$} ; 
\draw node[below] at (10,\hauteur) {$1$} ;  
\draw[fill=black] (10, 5) circle (0.05) ;

\draw[fill=blue7!75] (5,0) -- (10,5) -- (5,10) -- (0,5) -- (5,0) ;

\draw[fill=red7] (5,5) circle ({10*sqrt(1/2)/2}) ;



\draw[dashed] ($(5,0) + (65:9.08)$) arc (65:115:9.08) ;
\draw[dashed] ($(0,5) + (-25:9.08)$) arc (-25:25:9.08) ;
\draw[dashed] ($(5,10) + (65+180:9.08)$) arc (65+180:115+180:9.08) ;
\draw[dashed] ($(10,5) + (-25+180:9.08)$) arc (-25+180:25+180:9.08) ;

  \draw[fill=violet!90]
   (5,10)
   --
   (5 + 0.9175, 9.082)
   --
   (5 - 0.9175, 9.082)
  -- cycle ;
  
  \draw[fill=black]
  ($(5,0) + (84:9.08)$) arc (84:96:9.08)
   --
   (5 - 0.9175, 9.082)
  --
  (5 + 0.9175, 9.082)
  -- cycle ;

  \draw[fill=violet!90]
   (10, 5)
   --
   (9.082, 5 - 0.9175)
   --
   (9.082, 5 + 0.9175)
  -- cycle ;
  
  \draw[fill=black]
  ($(0,5) + (-6:9.08)$) arc (-6:6:9.08)
   --
   (9.082, 5 + 0.9175)
  --
  (9.082, 5 - 0.9175)
  -- cycle ;

  \draw[fill=violet!90]
   (5,0)
   --
   (5 + 0.9175, 0.9175)
   --
   (5 - 0.9175, 0.9175)
  -- cycle ;
  
  \draw[fill=black]
  ($(5,10) + (84+180:9.08)$) arc (84+180:96+180:9.08)
   --
   (5 + 0.9175, 0.9175)
  --
  (5 - 0.9175, 0.9175)
  -- cycle ;

  \draw[fill=violet!90]
   (0, 5)
   --
   (0.9175, 5 - 0.9175)
   --
   (0.9175, 5 + 0.9175)
  -- cycle ;
  
  \draw[fill=black]
  ($(10,5) + (-6+180:9.08)$) arc (-6+180:6+180:9.08)
   --
   (0.9175, 5 - 0.9175)
  --
  (0.9175, 5 + 0.9175)
  -- cycle ;


\newcommand{\mywidth}{0.3}
\newcommand{\myfactor}{3}

\draw[line width=\mywidth pt] (5, 9.3) -- (5 + 1.2, 9.3) -- (5+1.2, 8.8) -- (5, 8.8) -- cycle;
\draw[line width = \mywidth pt] (5 + 1.2, 9.3) -- (5 + 2, 9.8);
\draw[line width = \mywidth pt] (5 + 2, 9.8) -- (5 + 2, 9.8 + 0.5*\myfactor) -- (5+2 + 1.2*\myfactor, 9.8 + 0.5*\myfactor) -- (5+2 + 1.2*\myfactor, 9.8) -- cycle;

  \draw[fill=blue7!80, line width = \mywidth pt]
   (5 + 2, 9.8 + 0.5*\myfactor)
   --
   (5 + 2 + 1.2*\myfactor - 0.5*\myfactor, 9.8 + 0.5*\myfactor)
   --
   (5+2 + 1.2*\myfactor-0.015, 9.8)
   --
   (5+2, 9.8)
  -- cycle ;

  \draw[fill=violet!90, line width = \mywidth pt]
   (5 + 2, 9.8 + 0.5*\myfactor)
   --
   (5 + 2 + 1.2*\myfactor - 0.5*\myfactor, 9.8 + 0.5*\myfactor)
   --
   (5+2 + 0.9175*\myfactor, 9.8 + 0.282*\myfactor)
   --
   (5+2, 9.8 + 0.282*\myfactor)
  -- cycle ;
  
  \draw[fill=black, line width = \mywidth pt]
  ($(5+2,-14.5) + (83.35:9.8 + 0.282*\myfactor+14.5)$) arc (83.35:89:9.8 + 0.282*\myfactor+14.5)
   --
   (5+2, 9.8 + 0.282*\myfactor)
  --
  (5+2 + 0.9175*\myfactor, 9.8 + 0.282*\myfactor)
  -- cycle ;


\draw[fill=black] (5,10) circle (0.05) ;
\draw[fill=black] (5.1,10) node[right] {$\PR$};

\draw[fill=black] (10, 5) circle (0.05) ;
\draw[fill=black] (10.2, 5) node[above] {$\PRprime$};

\draw[fill=black] (5, 0) circle (0.05);
\draw[fill=black] (5.1, 0.) node[right] {$\PRbar$};

\draw[fill=black] (0, 5) circle (0.05) ;
\draw[fill=black] (-0.2, 5) node[above] {$\PRprimebar$};

\draw[fill=black] (5, 5) circle (0.05) node[above right] {$\I$};


\newcommand{\leftshift}{12.5}

\draw node[left] at (0.5+\leftshift,10) {\large(b)} ;

\draw[->] (5+\leftshift,-1) -- (5+\leftshift,11) ;
\draw node[anchor = east] at (4.7+\leftshift ,10.7) {$\sigma$};
\draw[->] (-0.5+\leftshift,5) -- (11+\leftshift,5) ;
\draw node[anchor = east] at (11+\leftshift, 5.5) {$\sigma'$};

\renewcommand{\hauteur}{5-0.12}
\renewcommand{\largeur}{5-0.05}

\draw node[left] at (\largeur+\leftshift, 0) {$-1$} ; 
\draw node[below] at (0+\leftshift, \hauteur) {$-1$} ; 
\draw[fill=black] (0+\leftshift, 5) circle (0.05) ;
\draw node[left] at (\largeur+\leftshift,10) {$1$} ; 
\draw node[below] at (10+\leftshift,\hauteur) {$1$} ;  
\draw[fill=black] (10+\leftshift, 5) circle (0.05) ;

\draw[fill=blue7!80] (5+\leftshift,0) -- (7.5+\leftshift, 7.5) -- (5+\leftshift,10) -- (2.5+\leftshift, 2.5) -- (5+\leftshift,0) ;

\coordinate (c1) at (5+\leftshift,5);
  \draw[fill=red7]
  ($(c1) + (45:35.355339059mm)$) arc (45:90:35.355339059mm)
  --
  ($(5+\leftshift,8) + (90:5.35mm)$) arc (90:155:6.80mm)
  -- 
  (4.17+\leftshift,7.5);

  \draw[fill=red7]
  ($(c1) + (225:35.355339059mm)$) arc (225:270:35.355339059mm)
  -- 
  ($(5+\leftshift,2) + (270:5.35mm)$) arc (270:335:6.80mm)
  --
  (5.83+\leftshift,2.5);

\draw[fill=red7, dash pattern = on 0pt off 2pt] (2.5+\leftshift,2.5) -- (5.83+\leftshift,2.5) -- (7.5+\leftshift,7.5) -- (4.17+\leftshift,7.5) -- (2.5+\leftshift,2.5) ;


\draw[dashed] (5+\leftshift,5) circle ({10*((3+sqrt(6))/6)-5}) ;

  \draw[fill=violet!90]
  (5+\leftshift,10)
  --
  (5 + 0.9175+\leftshift, 9.082)
   --
   (5 - 0.3+\leftshift, 9.082)
  -- cycle ;
  
  \draw[fill=black]
  ($(5+\leftshift, 5) + (75: 4.08)$) arc (75:94.4: 4.08)
   --
   (5 - 0.27+\leftshift, 9.082)
  --
  (5 + 0.9175+\leftshift, 9.082)
  -- cycle ;

  \draw[fill=violet!90]
  (5+\leftshift,0)
  --
  (5 - 0.9175+\leftshift, 0.9175)
   --
   (5 + 0.3+\leftshift, 0.9175)
  -- cycle ;
  
  \draw[fill=black]
  ($(5+\leftshift, 5) + (255: 4.08)$) arc (255:274: 4.08)
   --
   (5 + 0.27+\leftshift, 0.9175)
  --
  (5 - 0.9175+\leftshift, 0.9175)
  -- cycle ;

\draw[line width = 0.6pt, color = violet!90] (5+\leftshift, 10) -- (7.5+\leftshift, 7.5);
\draw[line width = 0.6pt, color = violet!90] (5+\leftshift, 0) -- (2.5+\leftshift, 2.5);
\draw (5+\leftshift, 0) -- (7.5+\leftshift, 7.5);
\draw (5+\leftshift, 10) -- (2.5+\leftshift, 2.5);


  \draw[fill=blue7!80, line width = \mywidth pt]
   (5+\leftshift + 2, 9.8 + 0.5*\myfactor)
   --
   (5+\leftshift + 2 + 1.2*\myfactor - 0.5*\myfactor, 9.8 + 0.5*\myfactor)
   --
   (5+\leftshift+2 + 1.2*\myfactor-0.015, 9.8)
   --
   (5+\leftshift+2, 9.8)
  -- cycle ;

  \draw[fill=violet!90, line width = \mywidth pt]
   (5+\leftshift + 2, 9.8 + 0.5*\myfactor)
   --
   (5+\leftshift + 2 + 1.2*\myfactor - 0.5*\myfactor, 9.8 + 0.5*\myfactor)
   --
   (5+\leftshift+2 + 0.9175*\myfactor, 9.8 + 0.282*\myfactor)
   --
   (5+\leftshift+2, 9.8 + 0.282*\myfactor)
  -- cycle ;
  
  \draw[fill=black, line width = \mywidth pt]
  ($(5+\leftshift+2,-1) + (74.1:9.8 + 0.282*\myfactor+1)$) arc (74.1:89:9.8 + 0.282*\myfactor+1)
   --
   (5+\leftshift+2, 9.8 + 0.282*\myfactor)
  --
  (5+\leftshift+2 + 0.9175*\myfactor, 9.8 + 0.282*\myfactor)
  -- cycle ;
  
  \draw[color = violet!90, line width = \mywidth*2.3 pt] 
   (5+\leftshift + 2 + 1.2*\myfactor - 0.5*\myfactor, 9.8 + 0.5*\myfactor)
   --
   (5+\leftshift+2 + 1.2*\myfactor-0.015, 9.8);
  
  \draw[line width=\mywidth pt] (5+\leftshift, 9.3) -- (5+\leftshift + 1.2, 9.3) -- (5+\leftshift+1.2, 8.8) -- (5+\leftshift, 8.8) -- cycle;
\draw[line width = \mywidth pt] (5 + 1.2+\leftshift, 9.3) -- (5+\leftshift + 2, 9.8);
\draw[line width = \mywidth pt] (5 + 2+\leftshift, 9.8) -- (5+\leftshift + 2, 9.8 + 0.5*\myfactor) -- (5+\leftshift+2 + 1.2*\myfactor, 9.8 + 0.5*\myfactor) -- (5+\leftshift+2 + 1.2*\myfactor, 9.8) -- cycle;


\foreach \Point/\PointLabel 
in { (5+\leftshift, 10)/\PR
       }
\draw[fill=black] \Point circle (0.05) node[above right] {$\PointLabel$};

\draw[fill=black] (7.5+\leftshift, 7.5) circle (0.05) node[below right] {$\SR$};

\draw[fill=black] (5+\leftshift, 0) circle (0.05);
\draw[fill=black] (5.3+\leftshift, 0.25) node[right] {$\PRbar$};

\draw[fill=black] (2.5+\leftshift, 2.5) circle (0.05);
\draw[fill=black] (2.5+\leftshift, 2.5) node[above left] {$\SRbar$};

\draw[fill=black] (5+\leftshift, 5) circle (0.05) node[above right] {$\I$};

\end{tikzpicture}